%% file: main.tex
\RequirePackage[T1]{fontenc}
\documentclass[letterpaper,10 pt,conference]{ieeeconf}
\input{preamble}
\title{\LARGE\bf Adaptive Event-Triggered Model Predictive Coverage Control with Average Update-Rate Regulation}
\author{Pengbo Zhu$^{1}$, Jicheng Shi$^{2}$, Colin N. Jones$^{2}$, and Andreas A. Malikopoulos$^{3}$%
\thanks{This research was supported in part by the National Science Foundation under Grants CNS-2401007, CMMI-2348381, and IIS-2415478; by MathWorks; and by the Swiss National Science Foundation (SNSF) through project P500-2\_235379 and NCCR Automation, a National Centre of Competence in Research, under Grant 51NF40\_180545.}%
\thanks{$^{1}$Pengbo Zhu is with the School of Civil \& Environmental Engineering, Cornell University, Ithaca, NY, USA (e-mail: \texttt{pz283@cornell.edu}).}%
\thanks{$^{2}$Jicheng Shi and Colin N. Jones are with the Automatic Control Lab, EPFL, Switzerland (e-mail: \texttt{jicheng.shi@epfl.ch}, \texttt{colin.jones@epfl.ch}).}%
\thanks{$^{3}$Andreas A. Malikopoulos is with Applied Mathematics, Systems Engineering, Mechanical Engineering, Electrical \& Computer Engineering, and the School of Civil \& Environmental Engineering, Cornell University, Ithaca, NY, USA (e-mail: \texttt{amaliko@cornell.edu}).}}
\begin{document}
\maketitle
\thispagestyle{empty}\pagestyle{empty}
\begin{abstract}
Model predictive coverage control for multi-agent systems operates under tight communication
and computation budgets, which requires it to balance the load of optimization against responsiveness to varying spatial priorities. 
With a fixed event threshold, the update rate varies with
operating conditions, making a prescribed average update budget difficult
to maintain. We propose threshold-adaptive event-triggered model
predictive control (MPC), which adjusts the threshold online in proportion
to the difference between the target average period and the latest
inter-event interval. We prove convergence of the empirical
average update period to any feasible prescribed target and establish
initialization-dependent one-sided bounds on the cumulative average.
A nonlinear benchmark illustrates these budget guarantees. We then
apply the mechanism to coverage MPC using a task-aware statistic that
captures robot motion and changes in spatial importance. In a five-agent
moving-density study, the adaptive policy uses \MovingSolves\ global MPC
updates on average against a target of 200, compared with \FixedSolves\
for a fixed-threshold policy using the same statistic and initial threshold.
It also reduces coverage cost relative to periodic schedules with exactly
matched update counts. These results demonstrate direct average-budget
regulation and effective task-dependent scheduling of MPC updates.
\end{abstract}
\input{introduction}
\input{method}

\input{case1}
\section{Case II: budgeted multi-agent coverage control}
\label{sec:coverage}
\input{coverage_method}
\input{coverage_figures}
\input{coverage_results}
\section{Conclusion}
\label{sec:conclusion}
Coverage control must respond to changing spatial demand while managing
the cost of repeated coordination and optimization. Event triggering
allows update times to follow task changes, while threshold adaptation
regulates their average frequency as operating conditions vary.
The proposed threshold-adaptive event-triggered MPC makes the desired
long-run average update period a direct design input, with asymptotic
tracking and initialization-dependent one-sided guarantees under the
stated assumptions, without requiring a disturbance distribution.
The nonlinear benchmark illustrates these timing properties. In the
moving-density coverage study, adaptation maintains the update count
near the requested budget when the density moves faster than in calibration, while
event-dependent timing improves coverage relative to periodic schedules
with matched update counts. The principal benefit is a direct update-budget
interface that retains the responsiveness of event-triggered MPC.
Future work will address stability-certified implementations,
asynchronous agent updates, and solver execution costs.

\useRomanappendicesfalse
\appendices

\section{Coverage objective and MPC architecture}\label{AppendixA}
At each global event, each agent solves a fixed-reference
tracking problem. For the sake of clarification, omitting the agent index, the objective is
\begin{equation}
\begin{aligned}
 \min_{\{x_j,u_j,\sigma_j\}}\quad&
 \sum_{j=0}^{N-1}
 \left(
 \|p_j-c_i\|_{Q_p}^2+\|v_j\|_{Q_v}^2
 +\|u_j\|_R^2+\rho\sigma_j^2
 \right)\\
 &+\|p_N-c_i\|_{Q_p}^2+\|v_N\|_{Q_v}^2
 +\rho\sigma_N^2,
\end{aligned}
\end{equation}
subject to $x_0=x_{k_i}$, the nominal dynamics in
\eqref{eq:double}, and
\begin{align}
 \|u_j\|_\infty&\le1, &&j=0,\ldots,N-1,\\
 \|v_j\|_\infty&\le1, &&j=0,\ldots,N,\\
 (0.5-\sigma_j)\mathbf{1}
 &\le p_j\le(9.5+\sigma_j)\mathbf{1},
 &&j=0,\ldots,N,\\
 \sigma_j&\ge0, &&j=0,\ldots,N.
\end{align}
Here $\|z\|_Q^2=z^\top Qz$,
$Q_p=2I$, $Q_v=0.5I$, $R=0.1I$, $\rho=200$,
and $N=N_{\max}=15$. And $c_i$ is
for each agent's event-time centroid. A scalar slack softens both position
coordinates at each stage, including the terminal stage.
The reference $c_i$ is held until the next global event.

\section{Coverage bound}\label{AppendixB}
 For the disturbance-free coverage example, the plant and predictor
use identical dynamics and inputs, giving $e_i^r(k)=0$ in exact
arithmetic. Let $a=\lceil\alpha\rceil$. Under the stated
componentwise velocity and input bounds, the displacement over
the first $a$ steps satisfies
\begin{equation}
 e_i^p(k_i+j)
 \le jT_s\sqrt{2}+\tfrac12 jT_s^2\sqrt{2},
 \qquad 1\le j\le a.
\end{equation}
If the centroids and fallback sites lie in $A=[0,10]^2$, their
separation is bounded by $\operatorname{diam}(A)=10\sqrt{2}$.
Consequently, a sufficient trigger bound is
\begin{equation}
 E_{\mathrm{cov}}
 =\max\left\{
 aT_s\sqrt{2}+\tfrac12 aT_s^2\sqrt{2},\,
 0,\,
 10\sqrt{2}
 \right\}.
\end{equation}
For $a=6$ and $T_s=0.2$ s, the displacement bound is approximately
$1.867$ m, so $E_{\mathrm{cov}}=10\sqrt{2}$ m. The initialization
bounds in Theorem~\ref{thm:budget} therefore become
\begin{align}
 L&=\eta(\alpha-N_{\max})=-0.09~\mathrm{m},\\
 U&=E_{\mathrm{cov}}+\eta(\alpha-1)
   \approx14.192~\mathrm{m}.
\end{align}
The value $U$ is a sufficient initialization for the one-sided
average-period guarantee, not the initialization used in the coverage experiments. The tested initial thresholds, $0.2125$ m
and zero, illustrate average-period regulation from general
initializations.

\bibliographystyle{IEEEtran}
\bibliography{refs}
\end{document}

%% file: preamble.tex
\IEEEoverridecommandlockouts
\usepackage{amsmath,amssymb}

\usepackage{amsthm,graphicx,algorithm,algorithmic}
\newtheorem{assumption}{Assumption}
\newtheorem{theorem}{Theorem}
\newtheorem{remark}{Remark}
\newtheorem{lemma}{Lemma}
\newtheorem{corollary}{Corollary}

\graphicspath{{figures/}}
\input{numbers}

\input{review_numbers}

\usepackage{xcolor}
\newcommand{\change}[1]{\textcolor{black}{#1}}

%% file: numbers.tex
\newcommand{\MovingSolves}{204.7}
\newcommand{\FixedSolves}{266.7}

\newcommand{\FixedGain}{3.14}

\newcommand{\ZeroSolves}{208.1}
\newcommand{\ZeroGain}{5.00}

%% file: review_numbers.tex
\newcommand{\AdaptiveCountError}{2.35}
\newcommand{\FixedCountError}{33.35}

%% file: introduction.tex
\section{Introduction}
\label{sec:intro}
\subsection{Background}
Coverage control coordinates mobile robots to distribute sensing or
service over an area according to spatial demand. A spatial importance
density assigns greater weight to regions requiring closer coverage;
Voronoi partitions and density-weighted centroids provide a classical
means of organizing the robots~\cite{cortes2004}. When the density varies,
the desired configuration changes with the task~\cite{lee2015}.
Maintaining coverage therefore requires refreshed position information,
spatial references, and control decisions. Frequent coordination can
burden communication, while repeated optimization adds computational
cost. Event-triggered control (ETC) schedules updates according to
observed need rather than a fixed clock~\cite{heemels2012}. Its relevance
to coverage is demonstrated by asynchronous event-triggered
broadcasting~\cite{ajina2021} and event- and self-triggered control and
communication for speed-constrained unicycle robots~\cite{zhou2021}.

\subsection{Related Work}
\subsubsection{Model predictive coverage control}
Model predictive control (MPC) incorporates robot dynamics and constraints
into the coverage problem, but each update entails solving an optimization
problem for each participating robot. Carron and
Zeilinger~\cite{carron2020mpcc} combine centroid-based coverage with
tracking MPC. Subsequent architectures address learning in unknown
environments~\cite{coverage} and multirate planning and tracking for
time-varying coverage~\cite{eberhard2025}. These developments make the
choice of when to refresh references and recompute control plans an
important part of resource-aware coverage.

\subsubsection{Resource budgets and event-triggered MPC}
This paper considers a prescribed average update period as a simple resource specification:
recompute the MPC plan once every $\alpha$ sampling steps on average,
where $\alpha\ge1$ is a user-specified target. This specifies a long-run
update rate rather than a hard limit on concurrent computations,
as considered in~\cite{grafe2025}. For a synchronously updated team
of $M$ agents, each global MPC update entails $M$ local solves.
The update count therefore provides a computational resource measure, and
total communication also depends on the implementation of trigger monitoring.

Event-triggered MPC recomputes a control plan when a triggering condition
is met, for example when a state-prediction error exceeds a
threshold~\cite{lehmann2013}. Adaptive triggering has also been studied
for perturbed nonlinear systems~\cite{aetmpc_li2023}. However, selecting
a threshold does not directly specify the resulting average update
period. A fixed threshold may require offline calibration, and the
resulting rate can change with the disturbance regime or task evolution.
Average-rate design has precedents: Brunner et al.~\cite{brunner2017}
develop robust event-triggered MPC with a guaranteed average sampling
rate, while Deng et al.~\cite{deng2022} use the disturbance distribution
to construct triggering sets for a prescribed expected inter-execution
time. Our objective is instead to regulate the realized empirical
average period through online threshold feedback, without requiring
a disturbance distribution.

The coverage control problem highlights why both task sensitivity and update rate
matter. A moving density can change the desired configuration even
when robots follow their predicted dynamics exactly; prediction error
alone may therefore miss a relevant task change. A task-aware statistic
can detect that change, but its growth rate also depends on robot motion
and density evolution. A threshold calibrated for one density speed
may consequently produce a different update rate at another speed.
Adapting the threshold from realized inter-event intervals addresses
this budget mismatch. Periodic scheduling can also meet an average
target, but event-based scheduling retains the freedom to place updates
when they are useful to the coverage task.

\subsection{Contributions and Organization}
We propose threshold-adaptive event-triggered MPC that adjusts its
threshold using the difference between the desired average period and
the latest inter-event interval. The feedback principle is related to
our disturbance-adaptive predictive control work~\cite{shi2024dadmpc,dadddpc},
which regulates average constraint violations; here, the regulated
quantity is the average inter-event period. Our contributions are:
\begin{enumerate}
\item 
A threshold-adaptation mechanism that regulates the average MPC update period to a prescribed target without
requiring a disturbance distribution. Under bounded-trigger and
executable-plan assumptions, we establish asymptotic average-period
tracking and initialization-dependent one-sided bounds.

\item A task-aware application to multi-agent coverage MPC, with a
sufficient trigger bound and a global MPC update-budget interpretation.
\item Complementary numerical evidence: a nonlinear benchmark illustrates
the timing guarantees, while moving-density coverage demonstrates budget
regulation under a task-speed change and improved coverage relative to
periodic schedules with matched update counts.
\end{enumerate}

The remainder of this paper is organized as follows.
Section~\ref{sec:prelim} presents the problem formulation, and
Section~\ref{sec:method} presents the method and its theoretical guarantees.
Section~\ref{sec:case1} evaluates the nonlinear benchmark.
Section~\ref{sec:coverage} presents the coverage formulation, the
triggering mechanism, and the associated case study.
Section~\ref{sec:conclusion} concludes the paper and outlines future work.

%% file: method.tex
\section{Preliminaries and problem formulation}
\label{sec:prelim}
Consider a discrete-time plant with bounded additive disturbance,
\begin{equation}
 x_{k+1}=f(x_k,u_k)+w_k,\qquad \|w_k\|_s\le\gamma,
 \label{eq:plant}
\end{equation}
with admissible states and inputs in $\mathcal X\subseteq\mathbb R^n$ and
$\mathcal U\subseteq\mathbb R^m$. The subscript $k$ indexes sampling instants,
whereas $i$ indexes events.
\change{At each event time $k_i$, an inner MPC generates an executable $N$-step input plan. Let $u_{k\mid k_i}$ denote the planned input for sampling instant $k$, generated at $k_i$. The corresponding nominal state prediction is initialized as $\hat{x}(k_i\mid k_i)=x_{k_i}$ and propagated as}
\begin{equation}
 \hat x(k+1\mid k_i)=f(\hat x(k\mid k_i),u_{k\mid k_i}).
\end{equation}
\change{Between events, the plant applies the stored plan, so $u_k=u_{k\mid k_i}$ for $k_i\leq k<k_{i+1}$.}
The original trigger statistic is the nominal prediction error
\begin{equation}
 e_i(k)=\|x_k-\hat x(k\mid k_i)\|_s\ge0.\label{eq:prediction}
\end{equation}
No new \change{MPC} solve
occurs while the old plan is being executed. The maximum allowed
inter-event interval \change{$N_{\max}$} satisfies $\lceil\alpha\rceil\le N_{\max}\le N$,
where $\alpha\ge1$ is the requested average period. \change{For convenience, define $a:=\lceil\alpha\rceil$}. 
The existence of
such plans throughout operation is an assumption on the inner MPC.

Let $k_0=0$ denote the mandatory initial event and
$\Delta_i=k_i-k_{i-1}$, $i\ge1$. Hence
$\sum_{j=1}^i\Delta_j=k_i$. Given $\alpha$, the three design aims are
\begin{align}
 \text{Aim 1: }&\lim_{i\to\infty} k_i/i=\alpha,\label{eq:aim1}\\
 \text{Aim 2: }&k_i/i\ge\alpha\quad\text{for all }i\ge1,\label{eq:aim2}\\
 \text{Aim 3: }&k_i/i\le\alpha\quad\text{for all }i\ge1.\label{eq:aim3}
\end{align}
The last two are alternative one-sided guarantees selected by
initialization. An average-period lower bound limits cumulative event
frequency from the initial time; it is not a minimum bound on every
individual interval or a budget for every sliding window.

\section{Threshold-adaptive event-triggered MPC}
\label{sec:method}
\subsection{Threshold adaptation and implementation}
Algorithm~\ref{alg:adaptive} implements the proposed threshold-adaptive
event-triggered MPC. In cycle $i$, it chooses
\begin{equation}
 \Delta_{i+1}=\min\bigl(\mathcal J_i\cup\{N_{\max}\}\bigr),
 \label{eq:trigger}
\end{equation}
where $\mathcal J_i=\{j\in\{1,\ldots,N_{\max}\}:
e_i(k_i+j)>\varepsilon_i\}$. It updates the threshold at the completed event:
\begin{equation}
 \varepsilon_{i+1}=\varepsilon_i+\eta(\alpha-\Delta_{i+1}),
 \qquad 0<\eta<\infty.\label{eq:update}
\end{equation}
A shorter-than-target interval raises the threshold; a longer one lowers it.
The maximum interval prevents exhaustion of the available plan. The
comparison is strictly $>$, not $\ge$. Negative thresholds are allowed:
they force a next-sample event, not a zero-time event. 

\begin{algorithm}[t]
\caption{Threshold-adaptive event-triggered MPC}\label{alg:adaptive}
\begin{algorithmic}[1]
\REQUIRE $\alpha\ge1$, $\eta>0$, finite $\varepsilon_0$;
 $a\le N_{\max}\le N$; inner MPC $\pi$.
\STATE Set $i=0$, $k_0=0$; measure $x_0$.
\LOOP
\STATE Solve $\pi(x_{k_i})$ for a length-$N$ input plan
 and its nominal state predictions.
\FOR{$j=1,\ldots,N_{\max}$}
\STATE Apply $u_{k_i+j-1\mid k_i}$; measure $x_{k_i+j}$.
\STATE Evaluate $e_i(k_i+j)$ using~\eqref{eq:prediction}.
\IF{$e_i(k_i+j)>\varepsilon_i$ \textbf{or} $j=N_{\max}$}
\STATE $\Delta_{i+1}=j$; $k_{i+1}=k_i+j$.
\STATE $\varepsilon_{i+1}=\varepsilon_i+\eta(\alpha-j)$.
\STATE $i\leftarrow i+1$; \textbf{break} the inner loop.
\ENDIF
\ENDFOR
\ENDLOOP
\end{algorithmic}
\end{algorithm}

\begin{remark}
Algorithm~\ref{alg:adaptive} only requires that $\pi(\cdot)$ produce a
valid open-loop input sequence of length at least $N_{\max}$. The
budget guarantee is therefore independent
of the MPC formulation. The MPC layer is responsible for state-constraint
satisfaction, recursive feasibility, and stability;
Algorithm~\ref{alg:adaptive} is responsible for the average period
budget.
\end{remark}

\subsection{Assumption and main theorem}
\begin{assumption}\label{ass:error}
For every admissible event cycle with the same applied inputs in plant
and predictor, there is a uniform finite $\varepsilon_{\max}\ge0$ such that
\begin{equation}
 0\le e_i(k_i+j)\le\varepsilon_{\max},\qquad 1\le j\le a.
 \label{eq:errorbound}
\end{equation}
Every executed update supplies an executable plan of length at least
$N_{\max}$, and these properties hold throughout operation.
\end{assumption}
Define the initialization thresholds
\begin{align}
 \varepsilon_{\min}&=\eta(\alpha-N_{\max})=:L\le0,\\
 \varepsilon_{\max,2}&=\varepsilon_{\max}+\eta(\alpha-1)=:U
 \ge\varepsilon_{\max}.\label{eq:upper}
\end{align} 

\begin{theorem}[Three budget guarantees]\label{thm:budget}
Under Assumption~\ref{ass:error}, $1\le\alpha\le N_{\max}\le N$,
and finite positive $\eta$, Algorithm~\ref{alg:adaptive} satisfies:
\begin{enumerate}
\item Aim 1 for every finite $\varepsilon_0$;
\item Aim 2 if $\varepsilon_0\ge U$;
\item Aim 3 if $\varepsilon_0\le L$.
\end{enumerate}
Moreover, for every event index $i\ge0$,
\begin{equation}
 \varepsilon_i\in[\min(\varepsilon_0,L),\max(\varepsilon_0,U)].
 \label{eq:bound}
\end{equation}
\end{theorem}
\begin{proof}
\emph{Part 1: reduce the three aims to threshold properties.}
Summing~\eqref{eq:update} and using $k_0=0$ yields
\begin{align}
 \varepsilon_i
 &=\varepsilon_0+\eta\sum_{j=1}^i(\alpha-\Delta_j)\nonumber\\
 &=\varepsilon_0+\eta(\alpha i-k_i).\label{eq:identity}
\end{align}
Thus, for $i\ge1$,
\begin{equation}
 \frac{k_i}{i}=\alpha+\frac{\varepsilon_0-\varepsilon_i}{\eta i}.
 \label{eq:average}
\end{equation}
Three sufficient properties follow. \textbf{T1:} a uniformly bounded
threshold sequence gives Aim 1. \textbf{T2:} if
$\varepsilon_i\le\varepsilon_0$ for every $i$, then Aim 2 holds.
\textbf{T3:} if $\varepsilon_i\ge\varepsilon_0$ for every $i$, then
Aim 3 holds. No statistical disturbance distribution is used.

\emph{Part 2: establish a forward-invariant threshold interval.}
Let $I=[L,U]$. Every interval is in $[1,N_{\max}]$. Consider three
regions of the real line.

\emph{Case 1: $\varepsilon_i\ge\varepsilon_{\max}$.}
By~\eqref{eq:errorbound}, strict signal crossing cannot occur during
the first $a$ steps. A forced event at $N_{\max}=a$ is possible, so
$\Delta_{i+1}\ge a\ge\alpha$. Consequently,
\begin{align}
 \varepsilon_{i+1}&\le\varepsilon_i,\nonumber\\
 \varepsilon_{i+1}&\ge\varepsilon_i+\eta(\alpha-N_{\max})
 \ge\varepsilon_{\max}+L\ge L.\label{eq:caseupper}
\end{align}

\emph{Case 2: $\varepsilon_i<0$.}
Nonnegativity gives $e_i(k_i+1)>\varepsilon_i$, hence
$\Delta_{i+1}=1$. Since $\alpha\ge1$,
\begin{equation}
 \varepsilon_i\le\varepsilon_{i+1}
 =\varepsilon_i+\eta(\alpha-1)<\eta(\alpha-1)\le U.
\end{equation}

\emph{Case 3: $0\le\varepsilon_i<\varepsilon_{\max}$.}
Using only $1\le\Delta_{i+1}\le N_{\max}$ gives
\begin{equation}
 L\le\varepsilon_{i+1}<\varepsilon_{\max}+\eta(\alpha-1)=U.
\end{equation}

If $\varepsilon_i\in I$, each case keeps $\varepsilon_{i+1}$ in $I$.
Above $U$, Case 1 makes the threshold non-increasing without crossing
below $L$. Below $L$, Case 2 makes it non-decreasing without crossing
above $U$. Therefore, the larger interval in~\eqref{eq:bound} is also
forward invariant, proving T1 for every finite initialization.
Finite-time entry into $I$ is not needed. For example, at $\alpha=1$
a negative threshold remains constant; the interval still equals one.

Uniform boundedness and~\eqref{eq:average} prove Aim 1. If
$\varepsilon_0\ge U$, then $[L,\varepsilon_0]$ is forward invariant:
Case 1 cannot increase the threshold and Cases 2--3 cannot exceed $U$.
Thus T2 proves Aim 2. If $\varepsilon_0\le L$, the symmetric invariant
interval $[\varepsilon_0,U]$ and T3 prove Aim 3.
\end{proof}

\begin{remark}[Finite-run update counts]
With $B=\max\{0,\varepsilon_0-L,U-\varepsilon_0\}/\eta$,
\eqref{eq:bound}--\eqref{eq:average} give
\begin{equation}
 |k_i-\alpha i|\le B,\qquad |k_i/i-\alpha|\le B/i\quad(i\ge1).
 \label{eq:finite}
\end{equation}
Define $S_T=\#\{i:k_i<T\}$, including the initial update but not
a terminal event at $T$. For the last update $k_i<T$, we have
$S_T=i+1$ and $0<T-k_i\le N_{\max}$; hence
\begin{equation}
 |S_T-T/\alpha|\le1+(B+N_{\max})/\alpha.
 \label{eq:count}
\end{equation}
Thus $S_T/T\to1/\alpha$, although a finite run need not contain
exactly $T/\alpha$ updates. The bound is conservative when the
signal bound is conservative.
\end{remark}

\subsection{Computing the prediction-error bound}
The following elementary recursion, used in robust MPC analyses such
as~\cite{limon2002}, can be used to verify Assumption~\ref{ass:error}.

\begin{lemma}[Prediction-error propagation]\label{lem:error}
Suppose $\|f(x,u)-f(y,u)\|_s\le L_f\|x-y\|_s$ on a domain containing
all relevant real and nominal trajectories, and $\|w_k\|_s\le\gamma$.
For equal initial states and the same input sequence,
\begin{equation}
 \|x_{k+j}-\hat x(k+j\mid k)\|_s
 \le E_j:=\gamma\sum_{r=0}^{j-1}L_f^r.\label{eq:lipschitz}
\end{equation}
\end{lemma}
\begin{proof}
Write $d_j=\|x_{k+j}-\hat x(k+j\mid k)\|_s$. The triangle
inequality and common-input Lipschitz bound give
$d_{j+1}\le L_fd_j+\gamma$, with $d_0=0$.
Induction yields~\eqref{eq:lipschitz}.
\end{proof}
Since $L_f\ge0$, $E_j$ is nondecreasing and one may take
\begin{equation}
 \varepsilon_{\max}=E_a=
 \begin{cases}
 \gamma(L_f^a-1)/(L_f-1),&L_f\ne1,\\
 a\gamma,&L_f=1.
 \end{cases}\label{eq:recipe}
\end{equation}
For differentiable $f$ on an appropriate convex state domain, a uniform
induced-norm bound on $\partial f/\partial x$ supplies $L_f$.
If $L_f<1$, $E_j\le\gamma/(1-L_f)$; global contraction is not required
because $a$ is finite. Bounds must hold on the prediction and actual
domains; imposing soft constraints does not by itself establish that fact.

%% file: case1.tex
\section{Case I: nonlinear budget verification}
\label{sec:case1}
We use a perturbed variant of the nonlinear benchmark
in~\cite{limon2002}:
\begin{align}
 x_{1,k+1}&=0.55x_{1,k}+0.12x_{2,k}\nonumber\\
 &\quad+(0.01-0.6x_{1,k}+x_{2,k})u_k+w_{1,k},\\
 x_{2,k+1}&=0.67x_{2,k}+(0.15+x_{1,k}-0.8x_{2,k})u_k+w_{2,k}.
\end{align}
The input satisfies $|u_k|\le0.1$, and $w_k=(0.1,0.1)$ on input
steps $k\bmod23<3$, zero otherwise. We set $N=N_{\max}=15$.
The zero-reference stage cost uses state/input weights $2I$ and 1;
The terminal
state weight is also $2I$.

The trigger is $e_i(k)=\|x_k-\hat x(k\mid k_i)\|_\infty$.
The maximum absolute row sum of the state Jacobian over
$|u|\le0.1$ is $0.85$, giving
$\varepsilon_{\max}=0.1(1-0.85^{\lceil\alpha\rceil})/(1-0.85)$
by Lemma~\ref{lem:error}. Runs last $T=1000$ steps, with
$\eta=0.01$, $\alpha\in\{2,5.5,8\}$, and
$\varepsilon_0\in\{L,0,U\}$. We use nine deterministic initial
states: $(-1.5,-1.5)$, $(0.9,-1.5)$, $(-1.5,1.5)$, $(0.9,1.5)$,
$(0.9,0)$, $(-1.5,0)$, $(0,1.5)$, $(0,-1.5)$, and $(0.6,0.8)$.
These 81 runs test the three aims. 
A further 324 runs scan
36 fixed thresholds, $0,0.01,\ldots,0.30$ and
$0.0999,0.1001,0.4,0.5,0.6$, at the same starts and inner MPC settings.

\begin{figure*}[tp]
\centering\includegraphics[width=.9\textwidth]{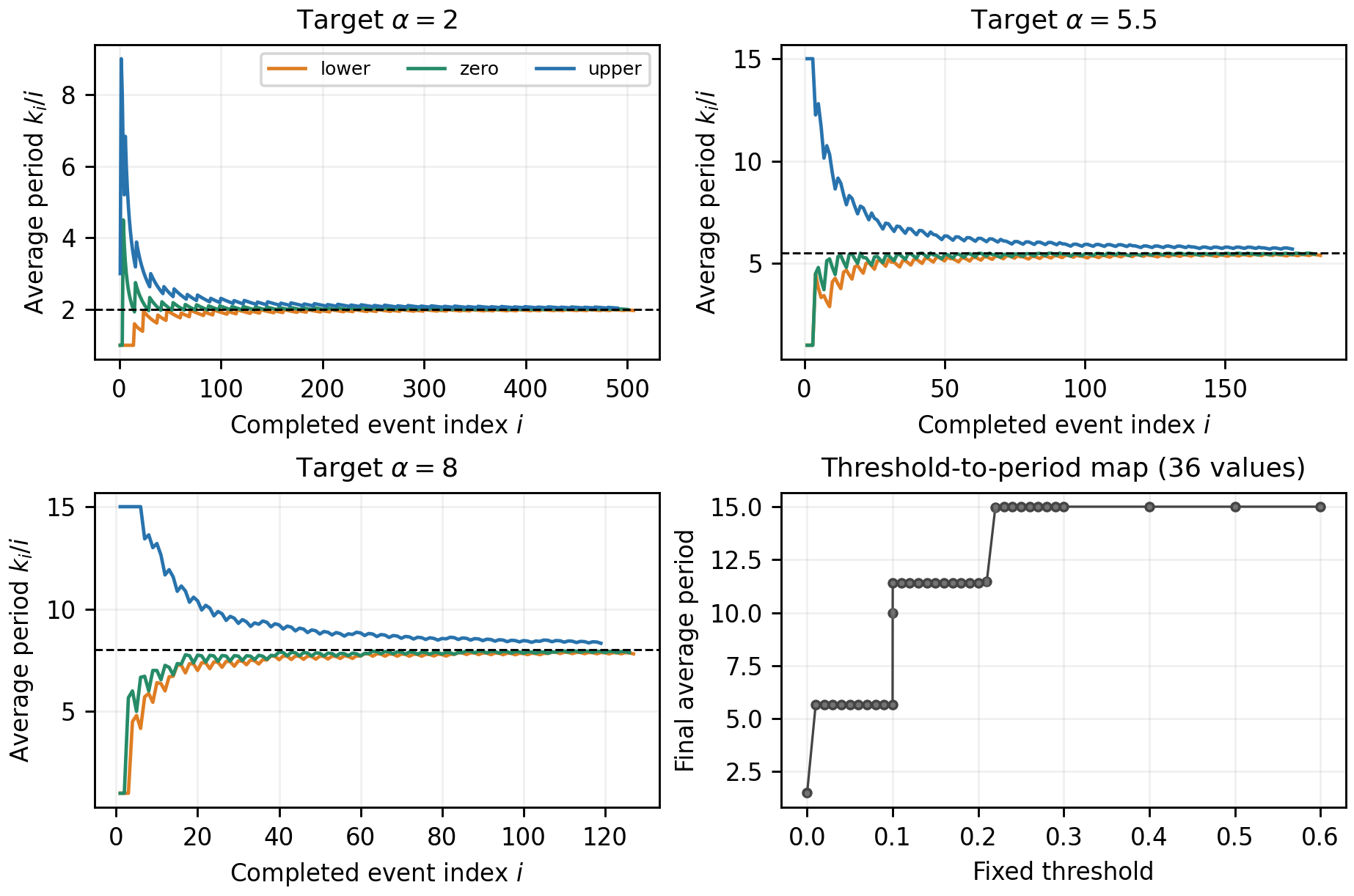}
\caption{Case I: average-period tracking and the sampled
fixed-threshold map. Thin lines show all nine starts; bold lines
use $(.9,0)$. Dashed lines are target periods. The rightmost panel
connects mean final periods at 36 tested thresholds, with
individual starts shown faintly.}
\label{fig:case1}
\end{figure*}

Figure~\ref{fig:case1} shows tracking from both prescribed sides
and from zero initialization, including the noninteger target
$\alpha=5.5$. All event-wise one-sided checks pass, and the
largest residual in \eqref{eq:identity} is
$4.4\times10^{-15}$. At $\alpha=5.5$, final average periods
are approximately $5.391$, $5.451$, and $5.701$ for lower,
zero, and upper initialization, respectively, illustrating
the finite-run transient around the asymptotic target.

The sampled threshold-to-period map shows pronounced plateaus:
selecting a threshold is not equivalent to selecting an average
period. In contrast, the proposed method accepts $\alpha$ directly without constructing
that map. 

%% file: coverage_method.tex
\subsection{Coverage objective and MPC architecture}
\label{sec:coverage_method}
Let $M$ agents with positions $p_m$ cover a compact convex region
$A\subset\mathbb R^2$ with nonnegative spatial importance
$\phi_k(q)$. The locational coverage cost is
\begin{equation}
 H(p,k)=\int_A\min_m\|q-p_m\|_2^2\phi_k(q)\,dq .
 \label{eq:coverage_cost}
\end{equation}
At a global event $k_i$, a central coordinator forms Voronoi cells
$W_{i,m}$ from current positions and computes
\begin{equation}
 c_m(W_{i,m},\phi_k)=
 \frac{\int_{W_{i,m}}q\phi_k(q)\,dq}
 {\int_{W_{i,m}}\phi_k(q)\,dq}.
 \label{eq:centroid}
\end{equation}
Building on the centroid-tracking MPC architecture of Carron and
Zeilinger~\cite{carron2020mpcc}, each
local tracking MPC receives $c_m(W_{i,m},\phi_{k_i})$ as its fixed
position reference and zero velocity reference. Plans, cells, and
references are held until the next global event. 
In~\cite{carron2020mpcc}, tracking MPC is solved every sample while
centroid references are updated conditionally. Here the adaptive trigger schedules
both reference refresh and replanning, using a simplified fixed-reference
MPC rather than the artificial-steady-state and terminal-set design
of that work (see details in Appendix A). Here $M=5$, $A=[0,10]^2$, $T_s=0.2$ s, and
\begin{equation}
\begin{aligned}
 p_{m,k+1}&=p_{m,k}+T_sv_{m,k}+\tfrac12T_s^2u_{m,k}+w^p_{m,k},\\
 v_{m,k+1}&=v_{m,k}+T_su_{m,k}+w^v_{m,k}.
\end{aligned}\label{eq:double}
\end{equation}
The predictor sets $w^p=w^v=0$. Each local MPC uses
$N=N_{\max}=15$, componentwise velocity/input limits
$1$ m/s and $1$ m/s$^2$, and stage weights $2I$, $0.5I$, and $0.1I$
for position error, velocity, and input. Terminal
position/velocity weights and velocity limits are identical. All policies use this same MPC, implemented
with CasADi/IPOPT (solver tolerance $10^{-6}$).

The plant dynamics are independent but references are spatially
coupled. A \emph{global MPC update} is one centrally coordinated event:
the cells and references are refreshed and all $M$ agents each solve
one local MPC problem. Thus $S_T$ global updates mean $MS_T$ local
solves (5 per update here). Between events, agents continue
applying their stored input sequences. 

\subsection{Task-aware triggering and its budget guarantee}
Moving density can make a reference ineffective even when the robot
follows its prediction exactly. We therefore use
\begin{align}
 e_i^p(k)&=\max_m\|p_m(k)-p_m(k_i)\|_2,\\
 e_i^r(k)&=\max_m\|p_m(k)-\hat p_m(k\mid k_i)\|_2,\\
 e_i^c(k)&=\max_m\|c_m(W_{i,m},\phi_k)
                     -c_m(W_{i,m},\phi_{k_i})\|_2,\\
 s_i(k)&=\max\{e_i^p(k),e_i^r(k),e_i^c(k)\}.
 \label{eq:task}
\end{align}
These are displacement, position-prediction difference, and centroid
drift (units: m). The centroid term uses event-held cells to
detect density changes within the current Voronoi partition. Current density
is known at each monitoring sample; future motion is not used.
Both adaptive and fixed policies use the same statistic.
We substitute $s_i$ for $e_i$ in Algorithm~\ref{alg:adaptive}.

\begin{corollary}[Budget regulation for the coverage trigger]\label{cor:signal}
If the nonnegative statistic~\eqref{eq:task} has a uniform bound
$E_{\rm cov}$ over the first $a$ steps of each cycle,
then Theorem~\ref{thm:budget} and \eqref{eq:count} apply to the
global update sequence with $\varepsilon_{\max}=E_{\rm cov}$,
under the same recursion, strict crossing, interval cap, and
executable-plan assumptions.
\end{corollary}
\begin{proof}
Identity~\eqref{eq:identity} is unchanged. The three-region proof
uses only the statistic's upper bound, nonnegativity, and interval
limits, respectively, so the same invariant interval and budget
conclusions hold with $s_i$ replacing prediction error.
\end{proof}

For the disturbance-free moving-density case, the plant and predictor
coincide under the same inputs, so $e_i^r=0$.
With $a=6$ and $\eta=0.01$, the sufficient
one-sided initializations are $L=-0.09$ m and $U\approx14.192$ m (see details in the Appendix B).

%% file: coverage_figures.tex
\begin{figure*}[tp]
\centering\includegraphics[width=\textwidth]{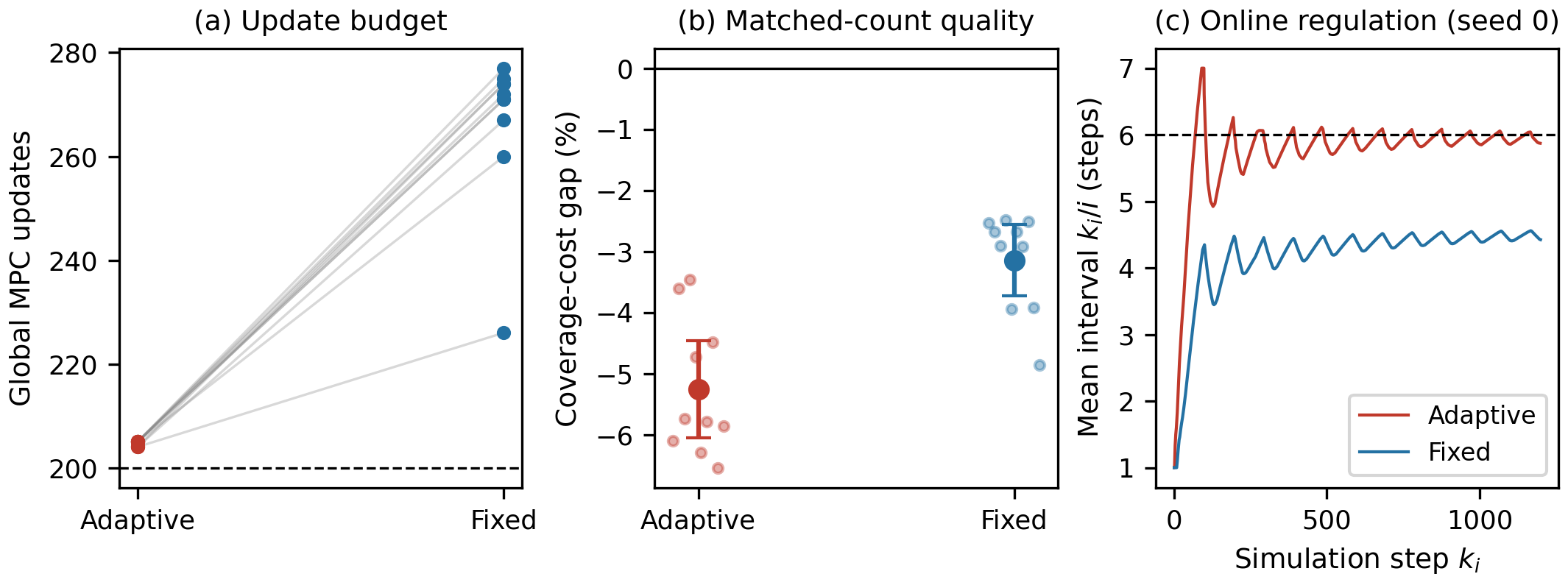}
\caption{Moving-density coverage with a common signal and initial
threshold. (a) Global MPC updates; gray lines pair ten initializations,
and the target is 200. One global update contains five local solves.
(b) Cost gaps against each policy's own exactly count-matched periodic
schedules: small points are seeds, large points are means, and bars
are 95\% confidence intervals. Negative gaps favor event-based timing;
the two columns have different baselines. (c) Cumulative average
interval $k_i/i$ at event times; the target is 6 steps.}
\label{fig:coverage}
\end{figure*}
\begin{figure*}[tp]
\centering\includegraphics[width=\textwidth]{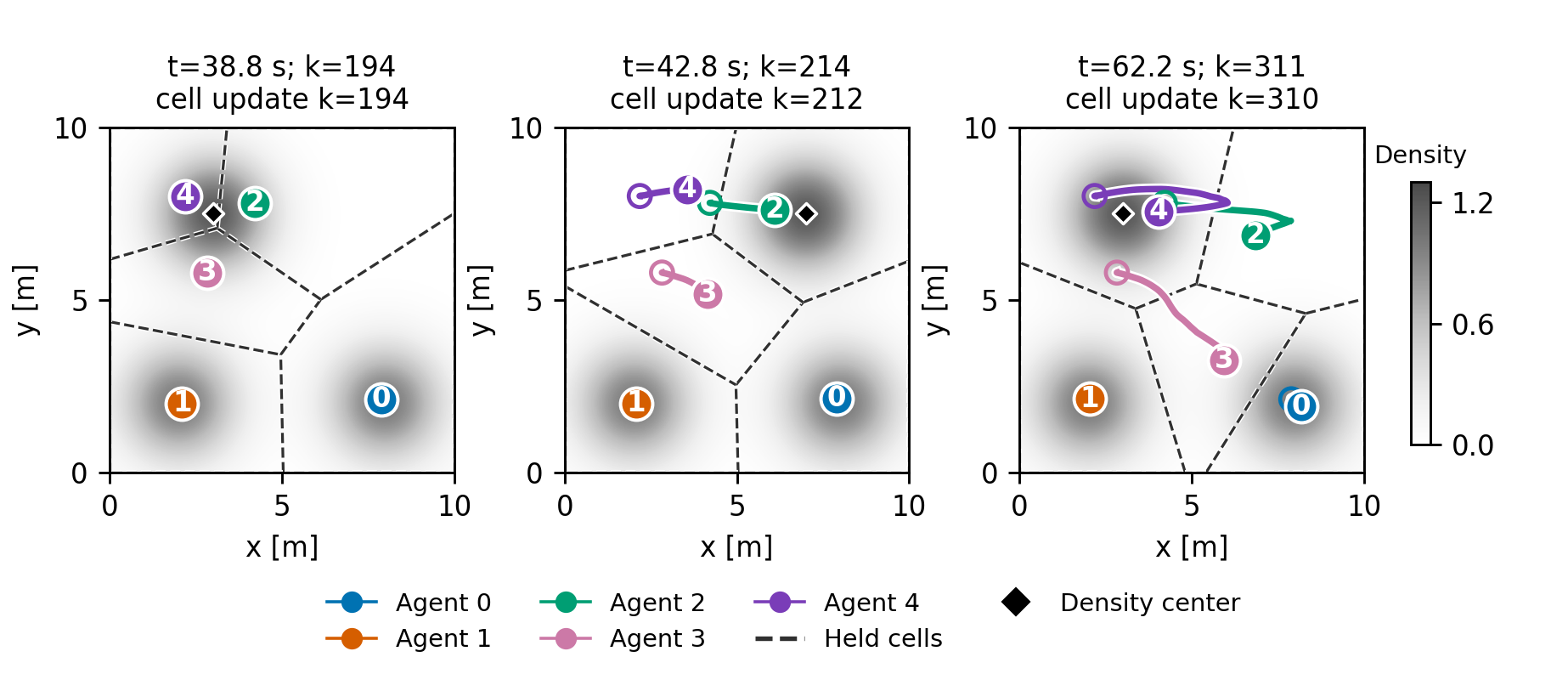}
\caption{Adaptive/task-aware coverage response. Trails
accumulate from $k=194$ to each displayed state
($t=38.8,42.8,62.2$ s) within a 240 s run. Numbered disks are
agents, open circles mark trail origins, and diamonds mark
Gaussian centers. Dashed Voronoi boundaries come from the latest
executed update at or before each state. The display grid is
$201\times201$; the evaluation grid remains $40\times40$.}
\label{fig:spatial}
\end{figure*}

%% file: coverage_results.tex
\subsection{Experimental setup}
The main density is a mixture of three isotropic Gaussians.
Two unit-amplitude components remain at $(2,2)$ and $(8,2)$;
a third, of amplitude 1.2, moves horizontally between $(3,7.5)$
and $(7,7.5)$. All standard deviations are 1.2 m. Beginning at
$k=0,97,194,\ldots$, motion occurs at 1 m/s for 20 steps, reversing
direction at each burst, with stationary intervals between bursts.
The main case has no plant disturbance. Runs last $T=1200$ steps
(240 s), with $\alpha=6$, $\eta=0.01$, initial positions uniform on
$[1.5,8.5]^2$, and zero initial velocities. Test seeds are 0 through
9, paired across policies.

The fixed threshold is calibrated offline using density speed
$0.6$ m/s. We select the threshold that
minimizes the mean absolute deviation of the global MPC update
count from the target $T/\alpha=200$. The search uses the
candidates $\{0,0.001,0.01,0.03,0.1,0.3,0.6,1\}$ m,
followed by four midpoint refinements around the target count.
This procedure gives $\varepsilon_0=0.2125$ m.

In the evaluation runs, both policies use the same density
speed of $1$ m/s, controller, trigger statistic, and initial
threshold $\varepsilon_0=0.2125$ m. The fixed policy retains
this threshold, whereas the adaptive policy updates it online.
The common signal, controller, and initialization isolate the
effect of adaptation under the task-speed change.

To isolate update placement, each policy/seed with $S$ actual
global MPC updates has its own integer periodic schedule with intervals
$d_\ell=\lfloor(\ell+1)T/S\rfloor-\lfloor\ell T/S\rfloor$,
$\ell=0,\ldots,S-1$. Cyclic rotations by
$0,\lfloor S/3\rfloor,\lfloor2S/3\rfloor$ are averaged within seed.
Each has exactly $S$ global updates, including $k=0$ but excluding $k=T$.
This count matching holds update quantity constant.
The full-run metric is
\begin{equation}
 J=\frac1T\sum_{k=1}^T
 \frac{H(p(k),k)}{\int_A\phi_k(q)\,dq},
 \label{eq:metric}
\end{equation}
evaluated on a $40\times40$ midpoint grid with the current-position
partition. This measures
density-weighted squared distance to the nearest agent over the
full run, rather than the local MPC tracking objective; lower is better.
For policy $\pi$ and seed $r$, the matched-count gap is
\begin{equation}
 g_{\pi,r}=100\left(\frac{J_{\pi,r}}
 {\bar J_{\rm periodic}(S_{\pi,r},r)}-1\right)\%.
 \label{eq:gap}
\end{equation}
The denominator averages the three periodic rotations with that
policy's realized update count. We report the mean gap and
paired-seed Student-$t$ 95\% confidence interval ($n=10$).

\subsection{Results and interpretation}
As shown in Fig.~\ref{fig:coverage}(a),
each point counts global MPC updates in one full run, and each gray line
pairs the same initial configuration across policies.
The target 200 corresponds to an average period of six sampling
steps (1.2 s), or 1000 individual MPC solves over 240 s.
The adaptive policy uses \MovingSolves\ global updates on average, whereas the
calibrated fixed policy uses \FixedSolves. Mean count deviations are
$\AdaptiveCountError\%$ and $\FixedCountError\%$.
The adaptive policy uses fewer global updates
relative to the fixed policy's mean count.
It illustrates the closeness to the requested budget under the
density-speed change, rather than fewer updates irrespective of
the target.

Periodic updates can meet an average budget, so the relevant question
is whether event-dependent timing uses the same number of updates
more effectively. In Fig.~\ref{fig:coverage}(b), against its own exact-count periodic comparator,
the adaptive policy achieves
lower cost in all ten runs.
The benefit therefore persists when the number of global updates,
and hence individual MPC solves, is exactly equal.
Fixed task-aware triggering improves
over its own periodic comparator by \FixedGain\%.
Each column uses its own matched-count baseline; the difference
between the two plotted mean gaps is not a direct adaptive--fixed
cost comparison. The test in Fig.~\ref{fig:coverage}(a) isolates budget adaptation at a common
signal, while the test in Fig.~\ref{fig:coverage}(b) demonstrates task-dependent update placement
within each realized budget. 

The red curve in Fig.~\ref{fig:coverage}(c) is the cumulative mean of completed intervals,
$k_i/i$, plotted against simulation step $k_i$ for one representative
initialization, not an average over seeds or an instantaneous interval.
It approaches the requested 6-step period, while the fixed policy
remains near 4.5 steps and therefore updates more frequently.
The remaining fluctuations are compatible with average-period
regulation: individual intervals can vary with task motion.
Together with the ten-seed counts in the test in Fig.~\ref{fig:coverage}(a) , this trajectory
illustrates both the transient adjustment and its finite-run outcome.

Figure~\ref{fig:spatial} shows the upper agents responding to density
motion while lower agents stay near the stationary components;
the trails display one response window within the 240 s run.
In ten additional runs with
$\varepsilon_0=0$, the adaptive policy uses \ZeroSolves\ global updates on average and
reduces cost by \ZeroGain\% against its own count-matched schedules. This supports operation
without scenario-specific threshold initialization in this case. 